\documentclass[letterpaper, 10pt, conference, twocolumn]{ieeeconf}

\IEEEoverridecommandlockouts
\usepackage{amsmath}
\usepackage{mathrsfs}
\usepackage{amsfonts}
\usepackage{graphicx}
\usepackage{amssymb}
\usepackage{latexsym}
\usepackage{array}
\usepackage{ifthen}
\usepackage{rotating}
\usepackage{morefloats}
\usepackage{enumerate}
\usepackage{bm}
\usepackage{stfloats}
\usepackage{dsfont}
\usepackage[english]{babel}
\usepackage{cite}
\usepackage{color}
\usepackage{flushend}
\usepackage {enumerate}
\newtheorem{theorem}{Theorem}
\newtheorem{lemma}{Lemma}

\newtheorem{example}{Example}
\newtheorem{conjecture}{Conjecture}

\title{\LARGE \bf Phase Analysis of MIMO LTI Systems \thanks{This work was supported
    in parts by the Research Grants Council of Hong Kong Special Administrative Region, China, under the Theme-Based Research Scheme T23-701/14-N.}}

\author{Wei Chen, Dan Wang, Sei Zhen Khong, and Li Qiu
\thanks{W. Chen, D. Wang, and L. Qiu are with the Department of Electronic and Computer Engineering, The Hong Kong University of Science and Technology, Clear Water Bay, Kowloon, Hong Kong, China. Email:
      wchenust@gmail.com, dwangah@connect.ust.hk, eeqiu@ust.hk}
\thanks{S. Z. Khong is with the Department of Electrical and Electronic Engineering, The University of Hong Kong, Pokfulam, Hong Kong, China. Email: szkhong@hku.hk}}

\begin{document}

\maketitle
\thispagestyle{empty}
\pagestyle{empty}

\begin{abstract}
In this paper, we introduce a definition of phase response for a class of multi-input multi-output (MIMO) linear time-invariant (LTI) systems, the frequency responses of which are cramped at all frequencies. This phase concept generalizes the notions of positive realness and negative imaginariness. We also define the half-cramped systems and provide a time-domain interpretation. As a starting point in an endeavour to develop a comprehensive phase theory for MIMO systems, we establish a small phase theorem for feedback stability, which complements the well-known small gain theorem. In addition, we derive a sectored real lemma for phase-bounded systems as a natural counterpart of the bounded real lemma.
\end{abstract}


\smallskip

\section{Introduction}

In the classical frequency domain analysis of single-input-single-output (SISO) systems, the magnitude (gain) response and phase response go hand in hand. In particular, the
Bode magnitude plot and phase plot are always drawn shoulder to shoulder. The combined Bode plot of a loop transfer function provides a significant amount of useful information about the closed-loop stability and performance. The gain and phase crossover frequencies of a loop transfer function give salient information on the
gain and phase margins of the feedback system. The famous Bode gain-phase integral relation binds the gain and phase together. In frequency domain controller synthesis, phase also plays an important role. Loop-shaping design techniques, such as lead and lag compensation, are rooted in the phase stabilization ideas.

The inception of MIMO systems theory sees extension and thriving of the magnitude concept, but not equal flourishing in the phase concept. While
the small gain theorem is widely known in the field of robust control, much less attention has been paid to the development of a small phase theorem. Moreover, the magnitude plot of a MIMO frequency response has been inbuilt to the computing environment MATLAB, a useful phase plot has not been available in practice. Several notable preliminary works on MIMO systems phases include \cite{Chen,Freudenberg,Anderson1988} and  \cite{Macfarlane1981}. The references \cite{Chen,Freudenberg,Anderson1988} extended the Bode gain-phase integral relation for SISO systems to MIMO systems. The reference \cite{Macfarlane1981} proposed a definition of phases for MIMO systems, based on which a small phase theorem was formulated. However, the condition therein depends on both phase and gain information, which somewhat deviates from the initial purpose of finding a phase counterpart to the small gain theorem.

An important line of research with a phasic point of view is on positive real (passive) and negative imaginary systems. Roughly speaking, one can think of positive real systems as those whose phases lie within $[-\frac{\pi}{2},\frac{\pi}{2}]$ and negative imaginary systems as those whose phases over positive frequencies lie within $[-\pi,0]$. Research on positive real systems can be traced back to more than half a century ago and
has led to a rich theory through efforts of generations of researchers. See books \cite{Anderson1973,BaoLee,Brogliato,DV1975} and the survey paper \cite{Kottenstette} for a review. Over the past two
decades, negative imaginary systems \cite{lanzon2008stability,petersen2010csm} and counter-clockwise dynamics \cite{Angeli2006} have attracted much attention. The abundant studies on these systems, concerning feedback stability, performance and beyond, provide valuable insights in developing a general phase theory for MIMO LTI systems.

One main reason accounting for the underdevelopment of MIMO phases is the following. While the gains of a complex
matrix are well described by its singular values, a universally accepted definition of
matrix phases has been lacking over a long period. Very recently, we initiated to adopt the canonical angles introduced in \cite{FurtadoJohnson2001} as the phases of a cramped complex matrix whose numerical range does not contain the origin \cite{WCKQ2018}. We studied various properties of matrix phases, some of which are briefly reviewed later. This paves the ground for conducting a systematic study of phase analysis and design for MIMO LTI systems.

In this paper, we first define the phase responses of MIMO LTI systems whose frequency responses are cramped at all frequencies. Such phase concept agrees with and generalizes the notions of positive realness and negative imaginariness. We then develop a small phase theorem for negative feedback interconnections of phase bounded systems, complementing the well known small gain theorem. We derive a sectored real lemma, which gives state space conditions for phase-bounded systems in terms of linear matrix inequalities (LMIs). This serves as a counterpart of bounded real lemma. In addition, we pay special attention to the class of half-cramped systems which exhibit a nice time-domain interpretation. We absorb much nutrition from the existing studies on positive real systems, negative imaginary systems, KYP lemma, generealized KYP lemma, integral quadratic constraints (IQCs), etc. along the way.

The rest of the paper is organized as follows. A review of matrix phases is presented in Section II. The phase responses of MIMO LTI systems are defined in Section III, followed by the discussions on half-cramped systems in Section IV. A small phase theorem is presented in Section V. State-space conditions are derived for phase bounded systems in Section VI. The paper is concluded in Section VII. The notation used in this paper is more or less standard and will be made clear as we proceed.

\section{Phases of a Complex Matrix}
A nonzero complex scalar $c$ can be represented in the polar form as $c=\sigma e^{i\phi}$ with $\sigma >0$ and $\phi$ taking values in a half open
$2\pi$-interval, typically $[0,2\pi)$ or $(-\pi,\pi]$.  Here $\sigma=|c|$ is called the modulus or the magnitude and $\phi=\angle c$ is called the
argument or the phase.  The polar form is particularly useful when multiplying two complex numbers. We simply have $|ab| =|a| |b|$ and
$\angle (ab) = \angle a + \angle b \mbox{ mod $2\pi$}$.

It is well understood that an $n \times n$ complex matrix $C$ has $n$ magnitudes, served by the $n$ singular values
\[
\sigma (C) = \begin{bmatrix} \sigma_1 (C) & \sigma_2 (C) & \cdots & \sigma_n (C) \end{bmatrix}
\]
with
$\overline{\sigma}(C)=\sigma_1 (C) \geq \sigma_2 (C) \geq \cdots \geq \sigma_n(C) = \underline{\sigma}(C)$ \cite{HornJohnson}. The magnitudes of a matrix possess plentiful nice properties, among which the following majorization inequality regarding the magnitudes of matrix products are of particular interest to the control community.

Given $x,y\in \mathbb{R}^n$, we denote by $x^\downarrow$ and $y^\downarrow$ the rearranged versions of $x$ and $y$ so that their elements
are sorted in a non-increasing order. Then, $x$ is said to be \emph{majorized} by $y$ \cite{Marshall}, denoted by $x\prec y$, if
\begin{align*}
  \sum_{i=1}^k x^\downarrow_i\leq\sum_{i=1}^k y^\downarrow_i,\ k=1,\ldots, n-1,\;\;\; \text{and} \;\;
  \sum_{i=1}^n x^\downarrow_i=\sum_{i=1}^n y^\downarrow_i.
\end{align*}
When $x$ and $y$ are nonnegative, $x$ is said to be \emph{log-majorized} by $y$, denoted by $x \prec_{\log} y$, if
\begin{align*}
  \prod_{i=1}^k x^\downarrow_i\leq\prod_{i=1}^k y^\downarrow_i,\ k=1,\ldots, n-1, \;\;\; \text{and} \;\;
  \prod_{i=1}^n x^\downarrow_i=\prod_{i=1}^n y^\downarrow_i.
\end{align*}
The magnitudes of matrix product satisfy \cite{Marshall}
\begin{align}
\sigma(AB) \prec_{\log} \sigma (A) \odot \sigma(B), \label{gainmajorization}
\end{align}
where $\odot$ denotes the Hadamard product, i.e., the elementwise product.

In contrast to the magnitudes of a complex matrix $C$, how to define the phases of $C$ appears to be an unsettled issue.  An early attempt \cite{Macfarlane1981} defined the phases of
$C$ as the phases of the eigenvalues of the unitary part of its polar decomposition. This definition was motivated by the seeming generalization of the polar form of a scalar to the polar decomposition of a
matrix. However, phases defined this way do not have certain desired properties.

Very recently, we discovered a more suitable definition of matrix phases based on numerical range \cite{WCKQ2018}. The numerical range, also called field of
values, of a matrix $C \in \mathbb{C}^{n\times n}$ is defined as $W(C) = \{ x^*Cx: x \in \mathbb{C}^n \mbox{ with } \|x\|=1\}$, which, as a subset of
$\mathbb{C}$, is compact and convex, and contains the spectrum of $C$ \cite{horntopics}.

If $0\notin W(C)$, then $W(C)$ is contained in an open half complex plane due to its convexity. In this case, $C$ is said to be a cramped matrix. It is known
that a cramped $C$ is congruent to a diagonal unitary matrix that is unique up to a permutation \cite{Horn,ZhangFuzhen2015}, i.e., there exists a nonsingular matrix
$T$ and a diagonal unitary matrix $D$ such that
$C=T^*DT$.
This factorization is called sectoral decomposition in \cite{ZhangFuzhen2015}.
Let $\delta(C)$ be the field angle of $C$, i.e., the angle subtended by the two supporting rays of $W(C)$ at the origin.
We define the phases of $C$, denoted by
$\phi_1(C),\phi_2(C),\dots,\phi_n(C)$, to be the phases of the eigenvalues of $D$, taking values in an interval $(\theta, \theta+\pi)$, where
$\theta\in[-\pi,\delta(C))$.  The phases defined in this fashion coincide with the canonical angles of $C$ introduced in \cite{FurtadoJohnson2001}.  Assume
without loss of generality that
\[
\overline{\phi}(C)=\phi_1(C)\geq \phi_2(C)\geq \dots\geq \phi_n(C)=\underline{\phi}(C).
\]
Moreover, define $\phi (C) = [ \phi_1 (C) \ \ \phi_2 (C) \ \ \cdots \ \ \phi_n(C) ]$.

The phases defined above admit the maximin and minimax expressions
\cite{Horn}:
\begin{equation*}
\begin{split}
\phi_i(C)&=\max_{\mathcal{M}: \mathrm{dim}\mathcal{M}=i}\min_{x\in \mathcal{M}, \|x\|=1} \angle x^*Cx\\&=\min_{\mathcal{N}: \mathrm{dim}\mathcal{N}=n-i+1}\max_{x\in \mathcal{N}, \|x\|=1}\angle x^*Cx.
\end{split}
\end{equation*}
In particular,
\begin{align*}
\overline{\phi}(C)&=\max_{x\in\mathbb{C}^n,\|x\|=1}\angle x^*Cx,\\
\underline{\phi}(C)&=\min_{x\in\mathbb{C}^n,\|x\|=1}\angle x^*Cx.
\end{align*}
A graphic interpretation of the phases is illustrated in \mbox{Fig. \ref{fig1}}. The two angles from the positive real axis to each of the two supporting rays of $W(C)$ are $\overline{\phi}(C)$ and $\underline{\phi}(C)$ respectively. The other phases of $C$ lie in between.
\begin{figure}[htb]
\centering
\includegraphics[scale=0.5]{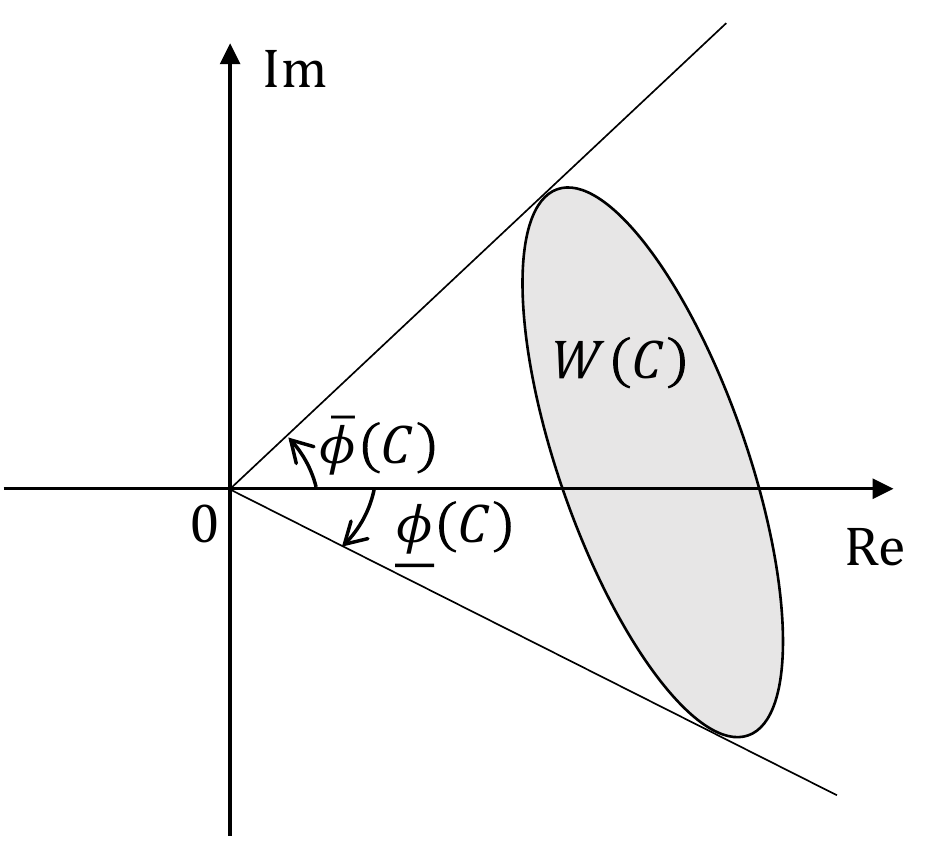}
\caption{Geometric interpretation of $\overline{\phi}(C)$ and $\underline{\phi}(C)$.}
\label{fig1}
\end{figure}

It is noteworthy that the notion of matrix phases subsumes the well-studied strictly accretive matrices \cite{Kato}, i.e., matrices with positive definite Hermitian part. In particular, the phases of $C$ lie in $(- \pi/2, \pi/2)$ if and only if $C$ is strictly accretive.

Given matrix $C$, we can check whether it is cramped or not by plotting its numerical range. From the plot of numerical range, we can also determine a $\pi$-interval $(\theta,\theta+\pi)$ in which the phases take values. How to efficiently compute $\phi(C)$ is an important issue. The following observation provides some insights along this direction. Suppose $C$ is cramped. Then it admits a sectoral decomposition $C=T^*DT$ and
thus
\begin{align*}
C^{-1}C^*=T^{-1}D^{-1}T^{-*}T^*D^*T=T^{-1}D^{-2}T,
\end{align*}
indicating that $C^{-1}C^*$ is similar to a diagonal unitary matrix. Hence, we can first compute $\angle\lambda(C^{-1}C^*)$, taking values in $(-2\theta\!-\!2\pi,-2\theta)$, and then let $\phi(C)=-\frac{1}{2}\angle\lambda(C^{-1}C^*)$. This gives one possible way to compute $\phi(C)$. We are currently exploring other methods, hopefully of lower complexity, for the computation of matrix phases.

The matrix phases defined above have plentiful properties, of which a comprehensive study has been conducted in \cite{WCKQ2018}.
First, note that the set of phase bounded matrices defined as
\begin{multline}
\mathcal{C}[\alpha, \beta]\\=\left\{C\!\in\! \mathbb{C}^{n\times n}: C \text{ is cramped and }\overline{\phi}(C)\!\leq\! \beta,\ \underline{\phi}(C)\!\geq\! \alpha\right\},\nonumber
\end{multline}
where $0\leq \beta-\alpha<2\pi$, is a cone. In addition, the following lemma can be shown by exploiting the maximin and minimax expressions of phases.

\begin{lemma}[\!\!\cite{WCKQ2018}]\label{convexcone}
If $\beta-\alpha<\pi$, then $\mathcal{C}[\alpha, \beta]$ is a convex cone.
\end{lemma}

Another important property pertinent to later developments in this paper is concerned with product of cramped matrices. In view of the magnitude counterpart in \eqref{gainmajorization}, one may expect
$\phi (AB) \prec \phi(A)+\phi(B)$ to hold for cramped matrices $A$ and $B$. This, unfortunately, fails even for positive definite $A$ and
$B$. Notwithstanding, if we consider instead $\lambda (AB)=\begin{bmatrix}\lambda_1(AB)&\dots&\lambda_n(AB)\end{bmatrix}$, i.e., the vector of eigenvalues of $AB$, the following weaker but useful result has been derived.

\begin{lemma}[\!\!\cite{WCKQ2018}] \label{thm: product_majorization} Let $A,B\in\mathbb{C}^{n\times n}$ be cramped matrices with phases in
  $(\theta_1, \theta_1+\pi)$ and $(\theta_2,\theta_2+\pi)$, respectively, where $\theta_1\in[-\pi,\delta(A))$ and $\theta_2 \in [-\pi,\delta(B))$. Let $\angle \lambda (AB)$ take
  values in $(\theta_1+\theta_2, \theta_1+\theta_2+2\pi)$. Then
\begin{align*}
\angle\lambda(AB) \prec \phi(A) + \phi(B).
\end{align*}
\end{lemma}

\vspace{3pt}
The above majorization relation underlies the development of a small phase theorem, much in the spirit of \eqref{gainmajorization} being the
foundation of the celebrated small gain theorem. To be more specific, recall that the singularity of matrix $I + AB$ plays an important role in
the stability analysis of feedback systems. It is straightforward to see that if $\sigma(A)$ and $\sigma(B)$ are both sufficiently small, then $I+AB$
is nonsingular. By contrast, one can observe that if
$\phi (A)$ and $\phi (B)$ are both sufficiently small in magnitudes, then $I+AB$ is nonsingular.

\section{Phase Response of MIMO LTI Systems}
Let $G$ be an $m\times m$ real rational proper stable transfer matrix, i.e.,
$G\in\mathcal{RH}^{m\times m}_\infty$. Then $\sigma(G(j\omega))$, the vector of singular values of $G(j \omega)$, is an $\mathbb{R}^m$-valued function
of the frequency, which we call the {\em magnitude response} of $G$. The $\mathcal{H}_\infty$ norm of $G$, denoted by
$\|G\|_\infty=\sup_{\omega \in \mathbb{R}} \overline{\sigma}(G(j \omega))$, is of particular importance.

Suppose $G(j \omega)$ is cramped for all $\omega \in \mathbb{R}$. Such a system is called a frequency-wise cramped system. Also, assume for simplicity that $W(G(j\omega))$ does not intersect the negative real axis for all $\omega \in \mathbb{R}$. Then
$\phi(G(j\omega))$, the vector of phases of $G(j\omega)$ with each element taking values in $(-\pi,\pi)$, is well defined as an $\mathbb{R}^m$-valued function of the frequency, which we call the {\em phase
response} of $G$. We define the $\mathcal{H}_\infty$ phase of $G$, as the counterpart to its $\mathcal{H}_\infty$ norm, to be
\begin{align*}
\Phi_\infty(G)= \displaystyle \sup_{\omega \in \mathbb{R}, \|x\| = 1} \angle x^* G(j\omega)x.
\end{align*}
Clearly, $\Phi_\infty(G)\!\leq\! \pi$.
It is noteworthy that the set of phase bounded systems
\begin{align}
\mathfrak{C}[\alpha]=\{G\in\mathcal{RH}_\infty^{m\times m}:\Phi_\infty(G)\leq \alpha\},\label{pbs}
\end{align}
where $\alpha\in[0,\pi)$, is a cone.

Having defined the phase response of $G$, we can now plot $\sigma(G(j\omega))$ and $\phi(G(j\omega))$ together to complete the MIMO Bode plot of $G$, laying the foundation of a complete MIMO frequency-domain analysis.
\begin{example}\label{example1}
The Bode plot of system
\begin{align*}
G(s)=\begin{bmatrix}\!\frac{1}{s^2+2s+200}\!&\frac{2}{s^2+2s+200}\!\vspace{3pt}
  \\\!\frac{2}{s^2+2s+200}\!&\frac{0.2s^3+0.5s^2+44.2s+24}{s^3+3s^2+202s+200}\!\end{bmatrix}
\end{align*}
is shown in Fig. \ref{mimo-response}.

\begin{figure}[htbp]
\centering
\includegraphics[scale=0.44]{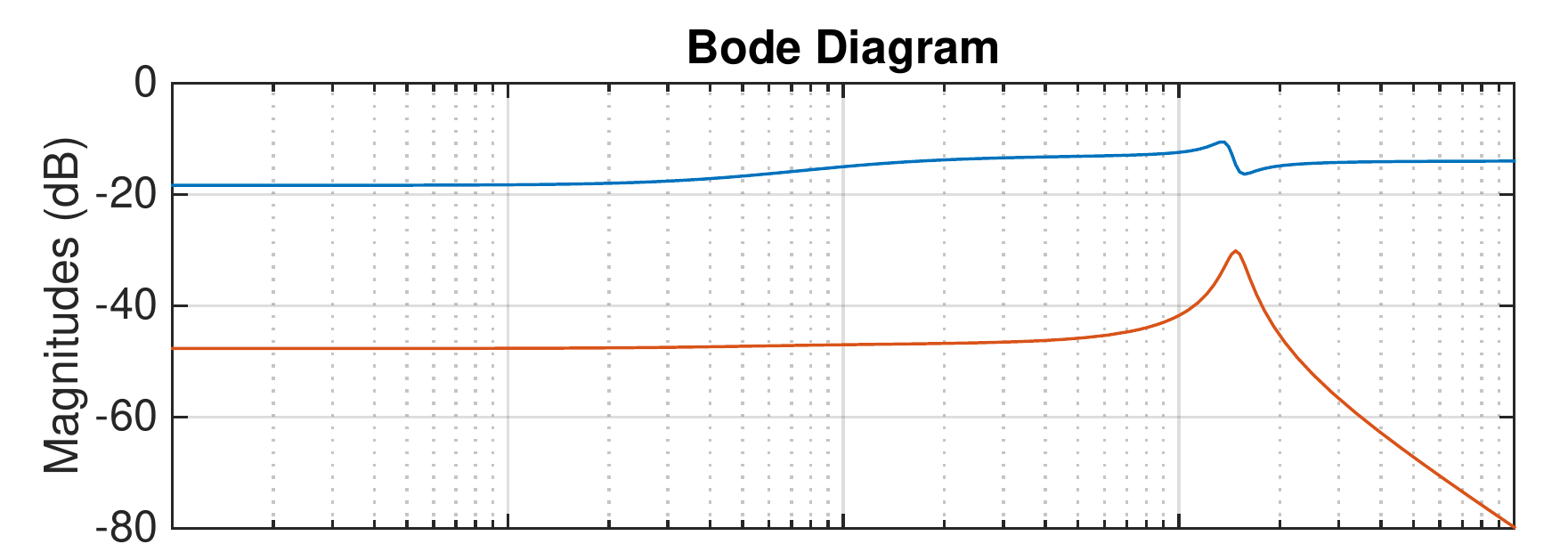}
\includegraphics[scale=0.44]{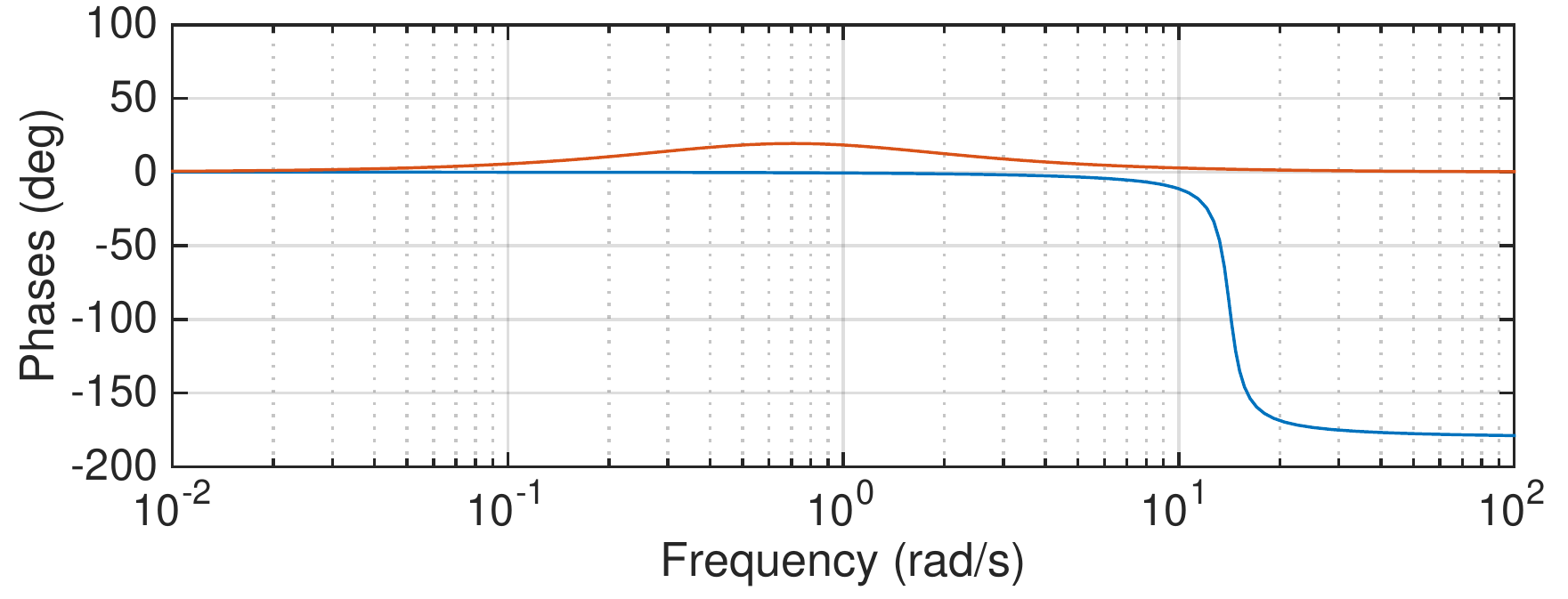}
\caption{MIMO Bode plot of a frequency-wise cramped system.}
\label{mimo-response}
\end{figure}
\end{example}

Note that the well-known notions of positive real systems \cite{Anderson1973,Brogliato,Kottenstette} and negative imaginary
systems \cite{lanzon2008stability,petersen2010csm} can be characterized using their phase responses. For simplicity, here we briefly mention the strong and strict
versions of these notions. A transfer function matrix $G \in \mathcal{RH}^{m\times m}_\infty$ is said to be strongly positive real if
$G(j\omega)+G^*(j\omega) > 0$ for all $\omega\in[-\infty,+\infty]$ \cite{LiuYao2016}. In the language of phase,
$G \in \mathcal{RH}^{m\times m}_\infty$ is strongly positive real if and only if
\[
\Phi_\infty (G) < \frac{\pi}{2}.
\]
On the other hand, a transfer function matrix $G$ is said
to be strictly negative imaginary if $(G(j\omega)-G^*(j \omega))/j< 0$ for all $\omega \in (0, \infty)$ \cite{lanzon2008stability}. This is equivalent
to
\[
[\underline{\phi}(G(j\omega)), \overline{\phi} (G(j \omega))] \subset (-\pi, 0)
\]
for all $\omega \in (0, \infty)$.
The phase concept of MIMO LTI systems gives a way to unify these concepts, together with of course the trivial SISO systems phase, and more. The system shown in Fig. \ref{mimo-response} is neither positive real nor negative imaginary but it has well-defined phase response.

\section{Half-cramped Systems}
Let $G\in\mathcal{RH}^{m\times m}_\infty$. Then, $G(j\omega)$ is conjugate symmetric, i.e.,
\[
G(-j\omega)=\overline{G(j\omega)},
\]
and hence $W(G(j\omega))$ and $W(G(-j\omega))$ are symmetric about the real axis. This property hints that in dealing with many problems such as feedback stability, one only has to examine the frequency response for nonnegative frequency, while the other half frequency range will be automatically taken care of due to the symmetry. Following this hint, we define half-cramped systems, and provide a time-domain interpretation for such systems.

A system $G$ is said to be half-cramped if
\[
\mathrm{cl.\;Co}\{W(G(j\omega)),\omega\geq 0\}
\]
is contained in an open half plane and does not intersect the negative real axis,
where cl. denotes closure and Co denotes convex hull.

Whether a system is half-cramped or not can be read out from its phase plot. For instance, the system in Example \ref{example1} is not half-cramped as its positive frequency phase response has a spread larger than $\pi$. Below we give an example of a half-cramped system.

\begin{example}
Consider the system
\begin{align*}
G(s)=\begin{bmatrix}\frac{s^3+6.5s^2+10s+6}{s^3+1.5s^2+1.5s+1}&\frac{s+2}{s+1}\vspace{3pt}\\\frac{s+2}{s+1}&\frac{s+2}{s+1}\end{bmatrix}.
\end{align*}
Its Bode plot is shown in Fig. 3, from which one can easily see that the system is half-cramped, but is neither positive real nor negative imaginary.
\begin{figure}[htbp]
\centering
\includegraphics[scale=0.44]{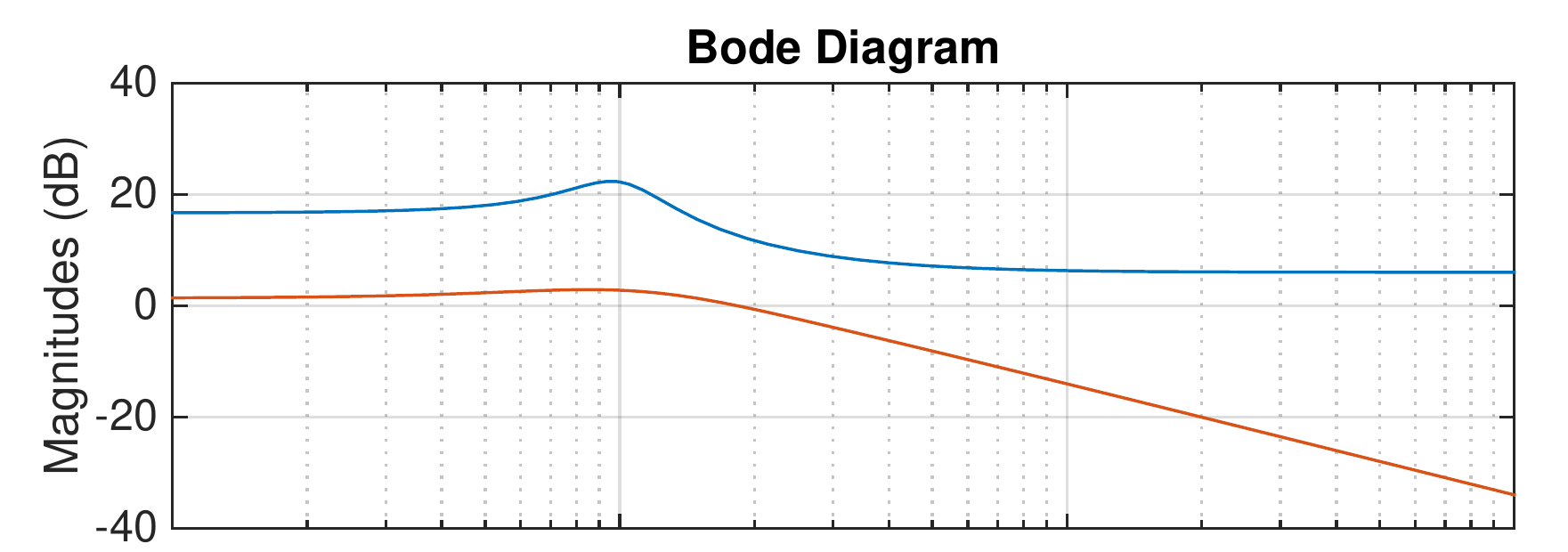}
\includegraphics[scale=0.44]{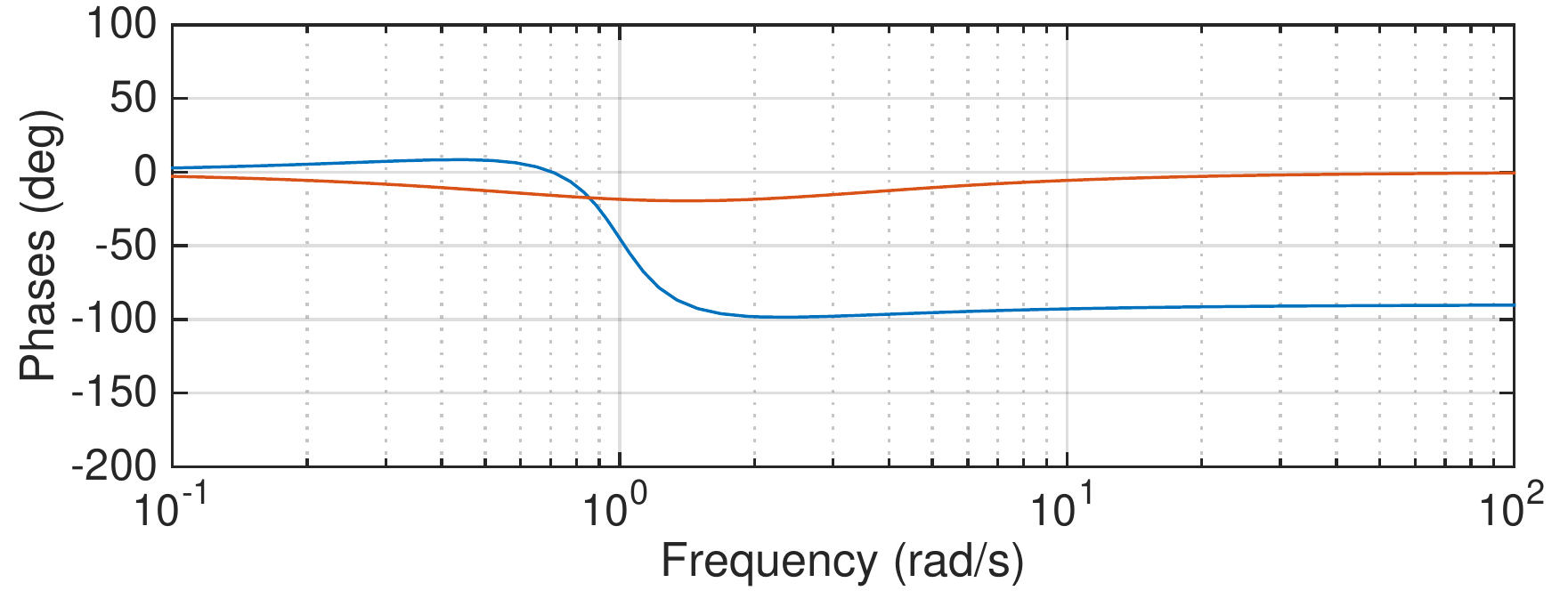}
\caption{MIMO Bode plot of a half-cramped system.}
\label{HC}
\end{figure}
\end{example}

Interestingly, there is a nice time-domain interpretation for half-cramped systems. For preparation, we briefly introduce some background knowledge on signal spaces and Hilbert transform. The Hilbert transform has been used extensively in signal processing, especially in the time-frequency domain analysis. It has also been applied in the control field, mostly in gain-phase relationship and system identification, etc. We refer interested readers to \cite{Hahn} for more details.

Let $\mathcal{F}$ be the usual Fourier transform on $\mathcal{L}^T_2(-\infty, \infty)$, the Hilbert space of complex-valued bilateral time functions
\[
[\mathcal{F} x ] (j\omega) = \frac{1}{\sqrt{2\pi}} \int_{-\infty}^{\infty} x(t) e^{-j\omega t} dt .
\]
Note that $\mathcal{F}$ is an isometry onto $\mathcal{L}^\Omega_2(-\infty, \infty)$, the Hilbert space of complex-valued bilateral frequency functions.
If we decompose $\mathcal{L}_2^\Omega (-\infty, \infty)$ into a positive frequency signal space and a negative frequency signal space as
\[
\mathcal{L}_2^\Omega(-\infty, \infty)= \mathcal{L}_2^\Omega(0, \infty) \oplus
 \mathcal{L}_2^\Omega(-\infty, 0),
\]
then clearly this is an orthogonal decomposition. Let $P$ be the orthogonal projection onto
$\mathcal{L}^\Omega_2 (0,\infty)$. Then we naturally have the orthogonal decomposition
\[
\mathcal{L}_2^T(-\infty, \infty)= \mathcal{F}^{-1} \mathcal{L}_2^\Omega(0, \infty) \oplus
\mathcal{F}^{-1} \mathcal{L}_2^\Omega(-\infty, 0).
\]
Let us call the first space above $\mathcal{A}$ and hence the second space $\mathcal{A}^\perp$. Let $Q$ be the orthogonal projection onto $\mathcal{A}$. Then the commutative diagram in Fig. \ref{commutativediagram} gives a complete picture of the relationships among these spaces. Recall the Hilbert transform $\mathcal{H}: \mathcal{L}^T_2(-\infty, \infty)  \rightarrow \mathcal{L}^T_2(-\infty, \infty)$ defined as
\[
[\mathcal{H} x](t)= \frac{1}{\pi}\int_{-\infty}^\infty \frac{x(\tau)}{t - \tau} \, d\tau.
\]
It then turns out that $Qx=\frac{1}{2} (x+j \mathcal{H} x)$ and $(I-Q)x=\frac{1}{2} (x-j \mathcal{H} x)$, the analytic part and the skew-analytic part of $x$ respectively.

\setlength{\unitlength}{0.98mm}
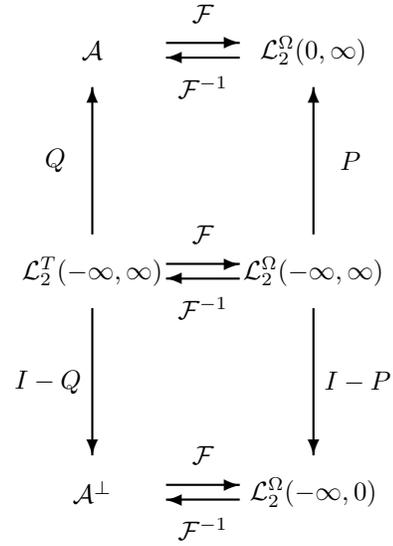
\begin{figure}[h!]
\begin{center}
\begin{picture}(40,70)
\thicklines
\put(25,64){\vector(-1,0){10}}
\put(15,66){\vector(1,0){10}}
\put(35,45){\makebox(10,10){$P$}}
\put(-5,45){\makebox(10,10){$Q$}}
\put(35,40){\vector(0,1){20}}
\put(30,60){\makebox(10,10){$\mathcal{L}^\Omega_2(0, \infty)$}}
\put(15,65){\makebox(10,10){$\mathcal{F}$}}
\put(15,55){\makebox(10,10){$\mathcal{F}^{-1}$}}
\put(0,60){\makebox(10,10){$\mathcal{A}$}}
\put(5,40){\vector(0,1){20}}
\put(0,30){\makebox(10,10){$\mathcal{L}^T_2(-\infty, \infty)$}}
\put(5,30){\vector(0,-1){20}}
\put(-6,15){\makebox(10,10){$I-Q$}}
\put(0,0){\makebox(10,10){$\mathcal{A}^\perp$}}
\put(25,34){\vector(-1,0){10}}
\put(15,36){\vector(1,0){10}}
\put(15,35){\makebox(10,10){$\mathcal{F}$}}
\put(15,25){\makebox(10,10){$\mathcal{F}^{-1}$}}
\put(25,4){\vector(-1,0){10}}
\put(15,6){\vector(1,0){10}}
\put(15,5){\makebox(10,10){$\mathcal{F}$}}
\put(15,-5){\makebox(10,10){$\mathcal{F}^{-1}$}}
\put(30,30){\makebox(10,10){$\mathcal{L}^\Omega_2(-\infty,\infty)$}}
\put(35,30){\vector(0,-1){20}}
\put(36,15){\makebox(10,10){$I-P$}}
\put(30,0){\makebox(10,10){$\mathcal{L}^\Omega_2(-\infty, 0)$}}
\end{picture}
\end{center}
\caption{A commutative diagram.}
\label{commutativediagram}
\end{figure}

Now let $\mathbf{G}: \mathcal{L}^T_2(-\infty, \infty)\rightarrow\mathcal{L}^T_2(-\infty, \infty)$ be the linear operator corresponding to $G(s)\in\mathcal{RH}_\infty$. Clearly, both $\mathcal{A}$ and $\mathcal{A}^\perp$ are invariant subspaces of $\mathbf{G}$. We define the positive frequency numerical range and negative frequency numerical range as
\begin{align*}
W_+(\mathbf{G})&:= \!\{\langle Qu, \mathbf{G}u \rangle:\ u \!\in \mathcal{L}_2^T (-\infty, \infty),\|u\|_2=1\}\\
W_-(\mathbf{G})&:= \!\{\langle (I-Q)u, \mathbf{G}u \rangle:\ u \!\in \mathcal{L}_2^T (-\infty, \infty),\|u\|_2=1\}
\end{align*}
respectively. It can be easily seen that $W_+(\mathbf{G})$ and $W_-(\mathbf{G})$ are symmetric with respect to the real axis.
Also, note that
\begin{align*}
\langle Qu, \mathbf{G}u \rangle&=\int_{-\infty}^{+\infty}[Qu]^*(t)[\mathbf{G}u](t) dt\\
&=\int_{0}^{+\infty} [\mathcal{F} u]^*(j\omega) G(j\omega)[\mathcal{F} u](j\omega)d\omega,
\end{align*}
which suggests that
\begin{align*}
\mathrm{cl.}\ W_+(\mathbf{G})\subset\mathrm{cl.\;Co}\{W(G(j\omega)),\omega\geq 0\}.
\end{align*}
In fact, one can further show
\begin{align*}
\mathrm{cl.}\ W_+(\mathbf{G})= \mathrm{cl.\;Co}\{W(G(j\omega)),\omega\geq 0\},
\end{align*}
and thus
$\Phi_\infty(G)=\max \{\sup_{z \in W_+(G)} \angle z , \sup_{z \in W_-(G)} \angle z \}$.
The detailed proof is omitted for brevity and will be available in a longer version of this paper.

\section{Small Phase Theorem}
Suppose $G$ and $H$ are $m\times m$ real rational proper transfer function matrices. The feedback interconnection of $G$ and $H$, as depicted in Fig. \ref{fdbk}, is said to
be stable if the Gang of Four matrix
\begin{align*}
G\#H=\begin{bmatrix}
  (I + HG)^{-1} & (I + HG)^{-1}H\\
  G(I + HG)^{-1} & G(I + HG)^{-1}H
\end{bmatrix}
\end{align*}
is stable, i.e., $G\#H \in \mathcal{RH}^{2m \times 2m}_\infty$.

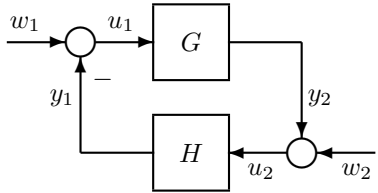
\begin{figure}[htbp]
\begin{center}
\begin{picture}(50,25)
\thicklines \put(0,20){\vector(1,0){8}} \put(10,20){\circle{4}}
\put(12,20){\vector(1,0){8}} \put(20,15){\framebox(10,10){$G$}}
\put(30,20){\line(1,0){10}} \put(40,20){\vector(0,-1){13}}
\put(38,5){\vector(-1,0){8}} \put(40,5){\circle{4}}
\put(50,5){\vector(-1,0){8}} \put(20,0){\framebox(10,10){$H$}}
\put(20,5){\line(-1,0){10}} \put(10,5){\vector(0,1){13}}
\put(5,10){\makebox(5,5){$y_1$}} \put(40,10){\makebox(5,5){$y_2$}}
\put(0,20){\makebox(5,5){$w_1$}} \put(45,0){\makebox(5,5){$w_2$}}
\put(13,20){\makebox(5,5){$u_1$}} \put(32,0){\makebox(5,5){$u_2$}}
\put(10,10){\makebox(6,10){$-$}}
\end{picture}
\caption{A standard feedback system.}
\label{fdbk}
\end{center}
\end{figure}

The celebrated small gain theorem \cite{Zhou,LiuYao2016} is one of the most used results in robust control theory over the past half a century. A version of it states that for $G, H \in \mathcal{RH}^{m\times m}_\infty$, the feedback system $G\#H$ is stable if
\begin{align*}
\overline{\sigma}(G(j\omega)) \overline{\sigma}(H(j\omega))< 1
\end{align*}
for all $\omega \in \mathbb{R}$.

There was an attempt to formulate a small phase theorem by using phases defined from the matrix polar decomposition \cite{Macfarlane1981}. However, the condition therein involves both phase and gain information and thus deviates from the initial purpose of having a phase counterpart of the small gain theorem.

Armed with the new definition of matrix
phases $\phi(C)$, we work out a version of the small phase theorem.

\begin{theorem}[Small phase theorem]\label{smallphase}
For frequency-wise cramped $G, H \in \mathcal{RH}^{m\times m}_\infty$, the feedback system $G\#H$ is stable if
\begin{align}
\overline{\phi} (G(j\omega))+ \overline{\phi}(H(j\omega)) < \pi \label{spi}
\end{align}
for all $\omega \in \mathbb{R}$.
\end{theorem}

\begin{proof}
Since $G, H \in \mathcal{RH}_\infty$, it follows that $G\#H$ is stable if and only if $(I + HG)^{-1}\in\mathcal{RH}_\infty$. Hence, it suffices to show that $\det[I+G(s)H(s)]\neq 0$ for all $s\in\mathbb{C}^+ \cup \{\infty\}$,
where $\mathbb{C}^+$ denotes the closed right half plane.

To this end, observe that when (\ref{spi}) is satisfied, by symmetry, the inequality
$\underline{\phi} (G(j\omega))+ \underline{\phi}(H(j\omega)) > -\pi$
also holds for all $\omega\in\mathbb{R}$. Applying Lemma  \ref{thm: product_majorization}, we have
\begin{multline}
\underline{\phi} (G(j\omega))+ \underline{\phi}(H(j\omega))\leq \angle\lambda_i(G(j\omega)H(j\omega))\\\leq \overline{\phi} (G(j\omega))+ \overline{\phi}(H(j\omega))\nonumber
\end{multline}
and thus
$-\pi<\angle\lambda_i(G(j\omega)H(j\omega))<\pi$
for all $\omega\in\mathbb{R},i=1,2,\dots,m$.
Now, let $\tau$ be an arbitrary number in $[0,1]$. From Lemma \ref{convexcone}, it follows that
\begin{align*}
&\overline{\phi}(\tau G(j\omega)+(1-\tau)I)\leq \overline{\phi}(G(j\omega)),\\
&\overline{\phi}(\tau H(j\omega)+(1-\tau)I)\leq \overline{\phi}(H(j\omega)),\\
&\underline{\phi}(\tau G(j\omega)+(1-\tau)I)\geq \underline{\phi}(G(j\omega)),\\
&\underline{\phi}(\tau H(j\omega)+(1-\tau)I)\geq \underline{\phi}(H(j\omega)),
\end{align*}
for all $\omega\in\mathbb{R}$. Then, following the same arguments as above, we can show
\begin{align*}
-\pi< \angle\lambda_i[(\tau G(j\omega)\!+\!(1-\tau)I)(\tau H(j\omega)\!+\!(1-\tau)I)]<\pi,
\end{align*}
which in turn yields that
\begin{align*}
\det[I+(\tau G(j\omega)\!+\!(1-\tau)I)(\tau H(j\omega)\!+\!(1-\tau)I)]\neq 0
\end{align*}
for all $\omega\in\mathbb{R}$. Since when $\tau=0$,
\begin{align*}
\det[I+(\tau G(s)\!+\!(1-\tau)I)(\tau H(s)\!+\!(1-\tau)I)]\neq 0
\end{align*}
for all $s\in\mathbb{C}^+$, it follows by continuity that the same holds for all $\tau\!\in\![0,1]$.
Particularly, when $\tau\!=\!1$, there holds $\det[I\!+\!G(s)H(s)]\neq 0$ for all $s\!\in\!\mathbb{C}^+$. Finally, note that $\det[I\!+\!G(\infty)H(\infty)]\neq 0$ due to the well-posedness of the feedback system. This completes the proof.
\end{proof}

We wish to mention that the small phase theorem can also be established via IQCs. Specifically, when the condition (\ref{spi}) is satisfied, one can find a dynamic multiplier of the form
\begin{align*}
\Pi(s)=\begin{bmatrix}0&e^{j\theta(s)}\\e^{-j\theta(s)}&0\end{bmatrix}
\end{align*}
so that $\Pi(s)\in \mathcal{L}_\infty$ is continuous on the imaginary axis and the following quadratic constraints
\begin{align*}
&\begin{bmatrix}I\\G(j\omega)\end{bmatrix}^*\Pi(j\omega)\begin{bmatrix}I\\G(j\omega)\end{bmatrix}\geq 0,\\
&\begin{bmatrix}H(j\omega)\\I\end{bmatrix}^*\Pi(j\omega)\begin{bmatrix}H(j\omega)\\I\end{bmatrix}<0
\end{align*}
are satisfied for all $\omega\!\in\!\mathbb{R}$. The feedback stability then follows from the result in \cite{MR}. From this perspective, the small phase theorem provides a nice phasic interpretation of the condition obtained from IQCs.

The small phase theorem generalizes a stronger version of the passivity theorem \cite{DV1975,LiuYao2016}, which states that
for $G, H \!\in\! \mathcal{RH}^{m\times m}_\infty$, the feedback system $G\#H$ is stable if $G$ and $H$ are strongly positive real.

Note that the small gain theorem provides a quantifiable tradeoff between the gains of $G$ and $H$, while the above small phase theorem does the same with respect to the phases of $G$ and $H$. In the literature, the notions of input feedforward passivity index and output feedback passivity index \cite{Vidyasagar,Wen,BaoLee,Kottenstette} have been used to characterize the tradeoff between
the surplus and deficit of passivity in open-loop systems. It is our belief that the concept of MIMO system phases is more suited to this
task. Specifically, $\frac{\pi}{2}-\Phi_\infty(G)$ gives a natural measure of passivity of system $G$, which we call the angular passivity index. The small phase theorem above implies that if the sum of the angular passivity indexes of $G$ and $H$ are positive, then $G\#H$ is stable. In addition, one can see that $\pi-\Phi_\infty(GH)$
yields a natural phase stability margin of $G\#H$.

It is well known that the condition given in the small gain theorem is necessary in the following sense \cite{Zhou}.
Suppose $G\!\in\!\mathcal{RH}^{m\times m}_\infty$ and let
$\mathfrak{B}[\gamma]\!=\!\{H\in\mathcal{RH}^{m\times m}_\infty: \!\|H\|_\infty \!\leq\! \gamma\}$, where $\gamma>0$. Then, the feedback system $G\#H$ is stable for all $H\in\mathfrak{B}[\gamma]$ if and only if $\|G\|_\infty < \frac{1}{\gamma}$.

Regarding the necessity of small phase theorem, we observe evidences supporting the following conjecture.
Recall the set of phase bounded systems $\mathfrak{C}[\alpha]$ defined in (\ref{pbs}), where $\alpha \in[0,\pi)$.
\begin{conjecture}[Small phase theorem with necessity]
Suppose $G\!\in\!\mathcal{RH}^{m\times m}_\infty$. Then, the feedback system $G\#H$ is stable for all $H\in\mathfrak{C}[\alpha]$ if and only if $\Phi_\infty(G)<\pi-\alpha$.
\end{conjecture}

Evidently, this conjecture holds in the SISO case in light of the Nyquist stability criterion.
A rigorous proof in the MIMO case appears technically challenging and is under our current investigation.

\section{State-space Conditions for Phase Bounded Systems}
The $\mathcal{H}_\infty$ norm of an LTI system can be determined by the well-known bounded real lemma. The efficient computation of $\mathcal{H}_\infty$ norm is specifically useful as evidenced in small gain theorem and facilitates robust control design.

The bounded real lemma \cite{Zhou} states that for $G\!\in\!\mathcal{RH}_\infty^{m\times m}$ with a minimal realization $\left[\begin{array}{c|c}A & B \\ \hline C & D\end{array} \right]$, $\|G\|_{\infty}\!<\!\gamma$ if and only if there exists $X>0$ satisfying the
  LMI
\begin{align*}
\begin{bmatrix}A'X+XA&XB&C'\\B'X&-\gamma I &D'\\C&D&-\gamma I\end{bmatrix}<0.
\end{align*}

One would naturally wish to see an analogous state-space condition for phase bounded systems. It is equally important to have an LMI characterization for a system $G$ satisfying $\Phi_\infty(G)\!<\!\alpha$, where $\alpha\in(0,\pi]$. Along this direction, we obtain a sectored real lemma, a natural counterpart of the bounded real lemma. Before proceeding, we introduce some preliminary knowledge on KYP lemma and generalized KYP lemma.

\subsection{KYP lemma and generalized KYP lemma}
The well known KYP lemma builds the equivalence between infinite many frequency domain inequalities over the entire frequency range and a finite dimensional LMI.

\begin{lemma}[KYP lemma \cite{LiuYao2016}]
Let $A\!\in\!\mathbb{C}^{n\times n}$, $B\!\in\!\mathbb{C}^{n\times m}$, $M\!=\!M^*\in\mathbb{C}^{(n+m)\times (n+m)}$. Assume that $A$ has no eigenvalues on the imaginary axis. Then the inequality
\begin{align*}
\begin{bmatrix}(j\omega I-A)^{-1}B\\I\end{bmatrix}^*M\begin{bmatrix}(j\omega I-A)^{-1}B\\I\end{bmatrix}<0
\end{align*}
holds for all $\omega\in\mathbb{R}\cup\{\infty\}$ if and only if there exists a Hermitian matrix $X$ satisfying the LMI
\begin{align*}
M+\begin{bmatrix}A^*X+XA&XB\\B^*X&0\end{bmatrix}<0.
\end{align*}
\end{lemma}

\vspace{8pt}

In contrast to the KYP lemma which copes with frequency domain inequalities over the entire frequency, the generalized KYP lemma \cite{IH} has the capability to address the frequency domain inequalities over partial frequency ranges.

Specifically, the generalized KYP lemma builds the equivalence between inequalities on curves in the complex plane and LMIs. Consider the curves characterized by the set
\begin{align*}
\mathrm{\bold{\Lambda}}(\Sigma,\Psi)=\left\{\lambda\in\mathbb{C}\left|\begin{bmatrix}\lambda\\1\end{bmatrix}^*\!\Sigma\begin{bmatrix}\lambda\\1\end{bmatrix}=0,\begin{bmatrix}\lambda\\1\end{bmatrix}^*\!\Psi\begin{bmatrix}\lambda\\1\end{bmatrix}\geq0\right.\right\},
\end{align*}
where $\Sigma,\Psi$ are given $2\times 2$ Hermitian matrices. By appropriately choosing $\Sigma$ and $\Psi$, $\mathrm{\bold{\Lambda}}(\Sigma,\Psi)$ can represent the partial or whole segment(s) of a straight line or a circle in the complex plane. When $\mathrm{\bold{\Lambda}}(\Sigma,\Psi)$ is unbounded, it is extended with $\infty$.

Denote by $\otimes$ the Kronecker product of matrices. A version of the generalized KYP lemma is as follows.

\begin{lemma}[Generalized KYP lemma\!\cite{IH}]
\mbox{Let \!$A\!\in\!\mathbb{C}^{n\times n}$}, $B\!\in\!\mathbb{C}^{n\times m}$, $M\!=\!M^*\!\in\!\mathbb{C}^{(n+m)\times(n+m)}$, and $\mathrm{\bold{\Lambda}}(\Sigma,\Psi)$ be curves in the complex plane. Let $\mathrm{\bold{\Omega}}$ be the set of eigenvalues of $A$ in $\mathrm{\bold{\Lambda}}(\Sigma,\Psi)$. Then the inequality
\begin{align*}
\begin{bmatrix}(\lambda I-A)^{-1}B\\I\end{bmatrix}^*M\begin{bmatrix}(\lambda I-A)^{-1}B\\I\end{bmatrix}<0
\end{align*}
holds for all $\lambda\in \mathrm{\bold{\Lambda}}(\Sigma,\Psi)\backslash \mathrm{\bold{\Omega}}$ if and only if there exist two Hermitian matrices $X$ and $Y$ such that
\begin{align*}
Y>0,\quad \begin{bmatrix}A&B\\I&0\end{bmatrix}^*(\Sigma\otimes X+\Psi\otimes Y)\begin{bmatrix}A&B\\I&0\end{bmatrix}+M<0.
\end{align*}
\end{lemma}

\vspace{8pt}

By choosing $\Sigma$ and $\Psi$ appropriately, one can use $\mathrm{\bold{\Lambda}}(\Sigma,\Psi)$ to define a variety of frequency ranges. For instance, when $\Sigma=\begin{bmatrix}0&1\\1&0\end{bmatrix},\Psi=0$, $\mathrm{\bold{\Lambda}}(\Sigma,\Psi)$ is simply the imaginary axis, and the generalized KYP lemma reduces to the classical KYP lemma.

\subsection{Sectored real lemma}
The following theorem gives a state space characterization for frequency-wise cramped systems satisfying $\Phi_\infty(G)<\alpha$, where $\alpha\in(0,\frac{\pi}{2}]$.
\begin{theorem}[Sectored real lemma]
Let $G\!\in\!\mathcal{RH}_\infty^{m\times m}$ with a minimal realization $\left[\begin{array}{c|c}A & B \\ \hline C & D\end{array} \right]$ and $\alpha\in(0,\frac{\pi}{2}]$. Then $\Phi_\infty(G)\!<\!\alpha$ if and only if there exists $X>0$ satisfying the LMI
\begin{align}
\begin{bmatrix}
A'X + XA &  XB \!-\! e^{-j(\frac{\pi}{2} - \alpha)}C' \\  B'X \!-\! e^{j(\frac{\pi}{2} - \alpha)}C & - e^{j(\frac{\pi}{2} - \alpha)}D\!-\!e^{-j(\frac{\pi}{2} - \alpha)}D'
\end{bmatrix} \!<\! 0.\label{srllmi}
\end{align}\label{SRL}
\end{theorem}

\begin{proof}
Note that $\Phi_\infty(G)<\alpha$, $\alpha\in(0,\frac{\pi}{2}]$ is equivalent to requiring $e^{j(\frac{\pi}{2}-\alpha)}G$ to be strongly positive real, i.e.,
\begin{align}
e^{j(\frac{\pi}{2}-\alpha)}G(j\omega)+e^{-j(\frac{\pi}{2}-\alpha)}G^*(j\omega)>0 \label{srli}
\end{align}
for all $\omega\in\mathbb{R}\cup\{\infty\}$. The inequality (\ref{srli}) can be rewritten as
\begin{align*}
\begin{bmatrix}(j\omega I\!-\!A)^{-1}B\\I\end{bmatrix}^*\!M\! \begin{bmatrix}(j\omega I\!-\!A)^{-1}B\\I\end{bmatrix}\!<\!0,
\end{align*}
where $M=\begin{bmatrix}0&-e^{-j(\frac{\pi}{2}-\alpha)}C'\\-e^{j(\frac{\pi}{2}-\alpha)}C&-e^{j(\frac{\pi}{2} - \alpha)}D-e^{-j(\frac{\pi}{2} - \alpha)}D'\end{bmatrix}$. Then, it follows from KYP lemma that $\Phi_\infty(G)<\alpha$ if and only if the LMI (\ref{srllmi}) has a Hermitian solution $X$. Finally, the positive definiteness of $X$ follows from the stability of $A$ and $A'X+XA<0$.
\end{proof}

When $\alpha=\frac{\pi}{2}$, the above sectored real lemma reduces to the strongly positive real lemma \cite{SKS1994}.

The case when $\alpha\!\in\!(\frac{\pi}{2},\pi]$ appears much more complicated. Nevertheless, for half-cramped systems, we are able to derive an LMI condition by employing the generalized KYP lemma.

In dealing with half-cramped real systems, one only needs to concern the frequency domain characterization for positive frequency, i.e., $\{j\omega|\omega\in[0,\infty]\}$. This frequency range can be captured by $\mathrm{\bold{\Lambda}}(\Sigma,\Psi)$ with $\Sigma=\begin{bmatrix}0&1\\1&0\end{bmatrix},\Psi=\begin{bmatrix}0&j\\-j&0\end{bmatrix}$.

Now we present the state-space condition for half-cramped systems.

\begin{theorem}
Let $G\in\mathcal{RH}_\infty^{m\times m}$ with a minimal realization $\left[\begin{array}{c|c}A & B \\ \hline C & D\end{array} \right]$ and $\alpha\!\in\!(\frac{\pi}{2},\pi]$. Then $G$ is half-cramped and $\Phi_\infty(G)\!<\!\alpha$ if and only if there exist Hermitian matrices $X$ and $Y$ satisfying either
\begin{align}
Y>0, \begin{bmatrix}A&B\\I&0\end{bmatrix}'\!\begin{bmatrix}0&X+jY\\X-jY&0\end{bmatrix}\!\begin{bmatrix}A&B\\I&0\end{bmatrix}+M<0,\label{LMI1}
\end{align}
where $M\!=\!\begin{bmatrix}0&-e^{-j(\alpha-\frac{\pi}{2})}C'\\-e^{j(\alpha-\frac{\pi}{2})}C&-e^{j(\alpha-\frac{\pi}{2})}D-e^{-j(\alpha-\frac{\pi}{2})}D'\end{bmatrix}$,
or
\begin{align*}
Y>0, \begin{bmatrix}A&B\\I&0\end{bmatrix}'\!\begin{bmatrix}0&X+jY\\X-jY&0\end{bmatrix}\!\begin{bmatrix}A&B\\I&0\end{bmatrix}+N<0,
\end{align*}
where $N\!=\!\begin{bmatrix}0&-e^{j(\alpha-\frac{\pi}{2})}C'\\-e^{-j(\alpha-\frac{\pi}{2})}C&-e^{-j(\alpha-\frac{\pi}{2})}D-e^{j(\alpha-\frac{\pi}{2})}D'\end{bmatrix}$.
\end{theorem}

\vspace{8pt}

\begin{proof}
By definition, we know $G$ is half-cramped and $\Phi_\infty(G)\!<\!\alpha$ for $\alpha\in(\frac{\pi}{2},\pi]$ if and only if either
\begin{align}
e^{j(\alpha-\frac{\pi}{2})}G(j\omega)+e^{-j(\alpha-\frac{\pi}{2})}G^*(j\omega)>0 \label{halfcrampedfdi}
\end{align}
or
\begin{align*}
e^{-j(\alpha-\frac{\pi}{2})}G(j\omega)+e^{j(\alpha-\frac{\pi}{2})}G^*(j\omega)>0
\end{align*}
holds for all $\omega\in[0,\infty]$. For brevity, we consider the case when (\ref{halfcrampedfdi}) holds for all $\omega\in[0,\infty]$. The other case can be shown similarly. The inequality (\ref{halfcrampedfdi}) can be rewritten into
\begin{align*}
\begin{bmatrix}(j\omega I-A)^{-1}B\\I\end{bmatrix}^*M\begin{bmatrix}(j\omega I-A)^{-1}B\\I\end{bmatrix}<0,\;\omega\in[0,\infty].
\end{align*}
Then, applying the generalized KYP lemma with $\Sigma\!=\!\begin{bmatrix}0&1\\1&0\end{bmatrix}$ and $\Psi=\begin{bmatrix}0&j\\-j&0\end{bmatrix}$ yields that the above frequency domain inequalities hold if and only if there exist Hermitian matrices $X$ and $Y$ satisfying LMIs (\ref{LMI1}). This completes the proof.
\end{proof}

\section{Conclusion} \label{ending}
In this paper, we define the phase responses of frequency-wise cramped MIMO LTI systems. The combined magnitude and phase plots constitute a complete MIMO Bode plot. We obtain a small phase theorem for closed-loop stability, a counterpart of the well-known small gain theorem. We also derive a sectored real lemma for phase-bounded systems, a counterpart of the bounded real lemma.

This paper focuses on the analysis of MIMO systems. We are currently working on the synthesis part, aiming at solving an $\mathcal{H}_\infty$-phase optimal control problem, a counterpart of the classical $\mathcal{H}_\infty$-norm optimal control problem. How to apply and extend phase analysis to large scale dynamical networks is also an interesting future work.

\end{document}